\documentclass[12pt,a4paper,final]{article}
\usepackage{color}
\newcommand{\ind}{1\!\!1_}
\newcommand{\dd}{\mathrm{d}}
\usepackage[a4paper,portrait]{geometry}
\usepackage{amsmath,amsfonts,latexsym,amssymb,amsthm}
\usepackage{graphicx}
\newcommand{\mykeywords}{Asymptotic Statistics. 
Covariance Estimation.
Time Series.}

\newtheorem{thm}{Theorem}[section]%
  \newtheorem{cor}[thm]{Corollary}%
  \newtheorem{prop}[thm]{Proposition}%
  \newtheorem{lem}[thm]{Lemma}%
  \newtheorem{hyp}[thm]{Assumption}%

\newcommand{\mysubjclass}{62G05, 62G20.}
\title{Estimation error for blind Gaussian time series prediction}
\author{T. Espinasse, F. Gamboa and  J-M. Loubes}

\begin{document}
\maketitle \noindent
\begin{abstract}

We tackle the issue of the blind prediction of a Gaussian time series. For this, we construct a projection operator build by plugging an empirical covariance estimation into a Schur complement decomposition of the projector. This operator is then used to compute the predictor. Rates of convergence of the estimates are given.

%
%
\end{abstract}
\tableofcontents
{ \noindent
$\;$ \\ \textbf{Subject Class. MSC-2000}: \mysubjclass}

\section*{Introduction}

In many concrete situations the statistician observes a finite path $X_1, \ldots, X_n$ of a real temporal phenomena which can be modeled as realizations of a stationary process ${\mathbf{X}}:=(X_t)_{t\in\mathbb{Z}}$ (we refer, for example, to \cite{davis}, \cite{Ste} and references therein).

Here we consider a second order weakly stationary process, which implies that its mean is constant and that $\mathbb{E}(X_tX_s)$ only depends on the distance between $t$ and $s$. In the sequel, we will assume that the process is Gaussian, which implies that it is also strongly stationary, in the sense that, for any $t,n \in \mathbb{Z}$, 
$$(X_1,\cdots,X_n)
\stackrel{\mathcal{L}}{=}
(X_{t+1},\cdots,X_{t+n}),\;(t\in\mathbb{Z},n\in\mathbb{N}).$$

Our aim is to predict this series when only a finite number of past values are observed. Moreover, we want a sharp control of the prediction error. For this, recall that, for Gaussian processes, the best predictor of $X_t, t \geq 0$, when observing $X_{-N}, \cdots, X_{-1}$, is obtained by a suitable linear combination of the $(X_i)_{i = -N, \cdots, -1}$. This predictor, which converges to the predictor onto the infinite past, depends on the unknown covariance of the time series. Thus, this covariance has to be estimated. Here, we are facing a blind filtering problem, which is a major difficulty with regards to the usual prediction framework. 

Kriging methods often impose a parametric model for the covariance (see \cite{krig2}, \cite{AzDa}, \cite{Ste}). This kind of spatial prediction is close to our work. Nonparametric estimation may be done in a functional way (see \cite{Lili10}, \cite{pred2}, \cite{pred4}). This approach is not efficient in the blind framework. Here, the blind problem is bypassed using an idea of Bickel \cite{bickel} for the estimation of the inverse of the covariance. He shows that the inverse of the empirical estimate of the covariance is a good choice when many samples are at hand.

We propose in this paper a new methodology, when only a path of the process is observed. For this, following Comte \cite{comte}, we build an accurate estimate of the projection operator. Finally this estimated projector is used to build a predictor for the future values of the process. Asymptotic properties of these estimators are studied.

The paper falls into the following parts. In Section \ref{s:notations}, definitions and technical properties of time series are given. Section \ref{s:frame} is devoted to the construction of the empirical projection operator whose asymptotic behavior is stated in Section \ref{s:rate}. Finally, we build a prediction of the future values of the process in Section \ref{section_schur}. All the proofs are gathered in Section \ref{s:append}.

\section{Notations and preliminary definitions}
\label{s:notations}

In this section, we present our general frame, and recall some basic properties about time series, focusing on their predictions.

Let $\mathbf{X} = (X_k)_{k\in\mathbb{Z}}$ be a zero-mean Gaussian stationary process.  Observing a finite past $X_{-N}, \cdots, X_{-1}$ ($N\geq 1$) of the process, we aim at predicting the present value $X_0$ without any knowledge on the covariance operator.

Since $X$ is stationary, let $r_{i-j }  := \operatorname{Cov}(X_i,X_j), (i,j \in \mathbb{Z}) $ be the covariance between $X_i$ and $X_j$.
Here we will consider short range dependent processes, and thus we assume that
$$\sum_{k \in \mathbb{Z}} r_k^2 < + \infty,$$
So that there exists a measurable function
$f^\star \in \mathbb{L}_2\left(\left[0,2\pi\right)\right)$ defined by 
\begin{equation*}
f^\star(t) : = \sum_{k= - \infty}^{\infty} r_{k}e^{ikt},  (a.e.)
\end{equation*} 

This function is the so-called spectral density of the time series. It is real, even and non negative. As $\mathbf{X}$ is Gaussian, the spectral density conveys all the information on the process distribution.\\

Define the covariance operator $\Gamma$ of the process $\mathbf{X}$, by setting
$$\forall i,j \in \mathbb{Z},\Gamma_{ij} = \operatorname{Cov}(X_i,X_j) .$$ 

Note that $\Gamma$ is the Toeplitz operator associated to $f^\star$. It is usually denoted by $T(f^\star)$ (for a thorough overview on the subject, we refer to \cite{bottcher}). 
This Hilbertian operator acts on $l^2(\mathbb{Z})$ as follows 
$$\forall u \in l^2(\mathbb{Z}), i \in \mathbb{Z}, (\Gamma u)_i := \sum_{j \in \mathbb{Z}} \Gamma_{ij}u_j  = \sum_{j \in \mathbb{Z}} r_{i-j}u_j = (T(f^\star)u)_i .$$
For sake of simplicity, we shall from now denote Hilbertian operators as infinite matrices. 

Recall that for any bounded Hilbertian operator $A$, the spectrum $\operatorname{Sp}(A)$ is defined as the set of complex numbers $\lambda$ such that 
$\lambda \operatorname{Id} - A$ is not invertible (here $\operatorname{Id}$ stands for the identity on $l^2(\mathbb{Z})$).

The spectrum of any Toeplitz operator, associated with a bounded function, satisfies the following property (see, for instance \cite{davis}):
$$\forall f \in \mathbb{L}_\infty \left(\left[0,2\pi\right)\right), \operatorname{Sp}(T(f)) \subset \left[\min(f), \max(f) \right] .$$

Now consider the main assumption of this paper :
\begin{hyp}\label{a:fbornee}
\begin{equation*}
\exists m,m' >0, \forall t \in \left[0,2\pi\right), m < f^\star(t) < m'.
\end{equation*}
\end{hyp} 

This assumption ensures the invertibility of the covariance operator,  since $f^\star$ is bounded away from zero.
As a positive definite operator, we can define its square-root $\Gamma^{\frac{1}{2}}$.
 Let 
 $Q$ be any linear operator acting on $l^2(\mathbb{Z})$, consider the operator norm $
\left\| Q \right\|_{2,op} :=  \sup_{u_\in l_2(\mathbb{Z}), \left\|u\right\|_2 = 1} \left\| Qu\right\|_2,$
and define the warped operator norm as 
\begin{equation*}
\left\| Q \right\|_{\Gamma} := \sup_{u_\in l_2(\mathbb{Z}), \left\|\Gamma^{\frac{1}{2}}u\right\|_2 = 1} \left\| \Gamma^{\frac{1}{2}}Qu\right\|_2.
\end{equation*} 
Note that, under Assumption \eqref{a:fbornee} $\left\| \Gamma\right\|_{2,op} \leq m'$, hence the warped norm $\left\| .  \right\|_\Gamma$ is well
 defined and equivalent to the classical one
\begin{equation*}
\frac{m}{m'} \left\| Q \right\|_{2,op} \leq \left\| Q \right\|_{\Gamma} \leq \frac{m'}{m}\left\| Q \right\|_{2,op}.
\end{equation*}

Finally, both the covariance operator and its inverse are continuous with respect to the previous norms.\\

\indent The warped norm is actually the natural inducted norm over the Hilbert space $$H = \left(l_2(\mathbb{Z}),\langle.,. \rangle_\Gamma \right),$$ 
where 
$$\langle x, y \rangle_\Gamma := x^T\Gamma y= \sum_{i,j \in \mathbb{Z}} x_i \Gamma_{ij} y_j.$$

From now on, all the operators are defined on $H$. 
Set $$\mathbb{L}^2(\mathbb{P}) := \left\{Y \in \overline{\operatorname{Span}}\left((X_i)_{i \in \mathbb{Z}}\right), \mathbb{E}[Y^2]<+\infty \right\}$$

The following proposition (see for instance \cite{davis}) shows the
particular interest of $H$ :
\begin{prop} The map 
\begin{align*}
\Phi : \quad H & \rightarrow \mathbb{L}^2(\mathbb{P}) \\
 u & \rightarrow u^T \mathbf{X}= \sum_{i \in \mathbb{Z}} u_i X_i  .
\end{align*}
defines a canonical isometry between $H$ and $\mathbb{L}_2(\mathbb{P})$.
 \end{prop}
The isometry will enable us to consider, in the proofs, alternatively sequences $u \in H$ or the corresponding random variables $Y  \in \mathbb{L}_2(\mathbb{P})$. \vskip .1in
We will use the following notations:  recall that $\Gamma$ is the covariance operator and denote, for any $A,B \subset \mathbb{Z}$,
the corresponding minor $(A,B)$ by 
$$\Gamma_{AB}:= \left(\Gamma_{ij}\right)_{i \in A, j \in B}.$$
 Note that, when $A$ and $B$ are finite, $\Gamma_{AB}$ is the covariance matrix between $(X_i)_{i \in A}$  and $(X_j)_{j \in B}$. Diagonal minors will be simply written $\Gamma_{A}: = \Gamma_{AA}$, for any $A \in \mathbb{Z}$.\\
\indent In our  prediction framework,  let $O \subset \mathbb{Z}$ and assume  that we observe the process $\mathbf{X}$ at times $i \in O $. It is well known that the best
linear prediction of a random variable $Y$ by observed variables 
$(X_i)_{i \in O}$ is also the best prediction, defined by $P_O(Y):=\mathbb{E}\left[ Y | (X_i)_{i \in O} \right]$.
 Using the isometry, there exist unique $u \in H$ and $v \in H$ with $Y=\Phi(u)$ and  $P_O(Y)=\Phi(v)$. 
Hence, we can define a projection operator acting on $H$, by setting $p_O(u):=v$. This corresponds to the natural projection in $H$ onto the set $\{ \Phi^{-1}(X_i),\: i \in O \}$.
Note that this projection operator may be written by block
$$p_O u := \left[\begin{matrix}
\Gamma_{O}^{-1}\Gamma_{O\mathbb{Z}} 
\\ 0            
           \end{matrix}\right]
u.$$
The operator $\Gamma_{O}^{-1}$ is well defined since $f^\star \geq m >0$. 
Finally, the best prediction observing $(X_i)_{i\in O}$ is  $$\mathbb{E}\left[ Y=\Phi(u) | (X_i)_{i \in O} \right]=P_O(\Phi(u)) = \Phi(p_O u).$$

This provides an expression of the projection when the covariance $\Gamma$ is known. Actually, in many practical situations, $\Gamma$ is unknown and need to be estimated from the observations. Recall that we observe $X_{-N},\dots,X_{-1}$. We will estimate the covariance with this sample and use a subset of these observations for the prediction. This last subset will be $\left\{(X_i)_{i \in O_{K(N)}} \right\} $, with $O_{K(N)}:=\left[-K(N), \cdots -1\right]$. Here $\left(K(N)\right)_{N \in \mathbb{N}}$ is a growing suitable sequence.
Hence, the predictor $\hat{Y}$ will be here
$$\hat{Y} = \hat{P}_{O_{K(N)}} Y,$$
where $\hat{P}_{O_{K(N)}}$ denotes some estimator of the projection operator onto $O_{K(N)}$, built with the full sample $(X_i)_{i = -N, \cdots, -1}$.

As usual, we estimate the accuracy of the prediction by the quadratic error
$${\rm MSE}(\hat{Y})=\mathbb{E}\left[ \left(\hat{Y} - Y\right)^2 \right].$$

The bias-variance decomposition gives 
$$\mathbb{E}\left[ \left(\hat{Y} - Y\right)^2 \right] = 
\mathbb{E}\big[ \left(\hat{P}_{O_{K(N)}} Y - P_{O_{K(N)}} Y \right)^2 \big]
   +  
\mathbb{E}\big[ \left(P_{O_{K(N)}}Y- P_{\mathbb{Z}^-} Y  \right)^2 \big]
+\mathbb{E}\big[ \left(P_{\mathbb{Z}^-} Y - Y \right)^2\big]    ,$$ 
where 
$$\hat{P}_{O_{K(N)}} Y =\hat{Y},$$
$$P_{O_{K(N)}} Y =\mathbb{E}\big[Y| (X_i)_{i \in O_{K(N)}} \big],$$
and
$$\hat{P}_{\mathbb{Z}^-} Y =\mathbb{E}\big[Y| (X_i)_{i <0} \big].$$
This error can be divided into three terms
\begin{itemize}
\item The last term $\mathbb{E}\big[ \left(P_{\mathbb{Z}^-} Y - Y \right)^2\big]$ is the prediction with infinite past error. It is induced by the variance of the unknown future values, and may be easily computed using the covariance operator. This variance does not go to zero as $N$ 
tends to infinity. It can be seen as an additional term that does not depend on the estimation procedure and thus will be omitted in the error term.
\item The second term $\mathbb{E}\big[ \left(P_{O_{K(N)}}Y- P_{\mathbb{Z}^-} Y  \right)^2 \big]$ is a bias induced by the temporal threshold on the projector.
\item  The first term $\mathbb{E}\big[ \left(\hat{P}_{O_{K(N)}} Y - P_{O_{K(N)}} Y \right)^2 \big]$ is a variance, due to the fluctuations of the estimation, and decreases to zero as soon as the estimator is consistent. Note that to compute this error, we have to handle the dependency between the prediction operator and the variable $Y$ we aim to predict.
\end{itemize}
Finally, the natural risk is obtained by removing the prediction with infinite past error: 
\begin{align*}R(\hat{Y} = \hat{P}_{O_{K(N)}} Y ) &: =\mathbb{E}\big[ \left(\hat{P}_{O_{K(N)}} Y - P_{O_{K(N)}} Y \right)^2 \big]
   +  
\mathbb{E}\big[ \left(P_{O_{K(N)}}Y- P_{\mathbb{Z}^-_*} Y  \right)^2 \big]
\\ & =\mathbb{E}\left[ \left( \hat{Y}-\mathbb{E}\left[Y|(X_i)_{i <0}\right] \right)^2 \right].\end{align*}

The global risk will be computed by taking the supremum of $R(\hat{Y})$ among of all random variables $Y$ in a suitable set (growing with $N$). This set
will be defined in the next section.
\section{Construction of the empirical projection operator}
\label{s:frame}


%

Recall that the expression of the empirical unbiased covariance estimator is given by (see for example \cite{AzDa})
\begin{equation*}
\forall ~0<p< N,~ \hat{r}^{(N)}(p) = \frac{1}{N-p}\sum_{k = -N}^{-p-1}X_k X_{k+p} .
\end{equation*}

Notice that, when $p$ is close to $N$, the estimation is hampered since we only sum $N-p$ terms. Hence, we will not use the complete available data but 
rather use a cut-off. 



Recall that $O_{K(N)} := [-K(N),-1]$ denotes the indices of the subset used for the prediction step.
We define the empirical spectral density as
\begin{equation} \label{specdens}
\hat{f}_K^{(N)} (t)= \sum_{p = -K(N)}^{K(N)} \hat{r}^{(N)}(p) e^{ipt}.
\end{equation}

We now build an estimator for $p_{O_{K(N)}}$ (see Section \ref{s:notations} for the definition of $p_{O_{K(N)}}$).

%
First, we divide the index space $\mathbb{Z}$ into $M_K \cup O_K \cup B_K \cup F_K$ where :
\begin{itemize}
\item $M_K = \left\{\cdots, -K-2, -K-1\right\}$ denotes the index of the  past data that will not be used for the prediction (missing data)
\item $O_K = -K, \cdots, -1$ the index of the data used for the prediction (observed data)
\item $B_K = 0, \cdots, K -1$ the index of the data we currently want to forecast (blind data)
\item $F_K = K, K+1, \cdots$ the remaining index (future data)
\end{itemize}
In the following, we omit the dependency on $N$ to alleviate the notations.

As discussed in Section \ref{s:notations}, the projection operator $p_{O_K}$ may be written by blocks as:

\begin{equation*}
p_{O_K} = \left[ \begin{matrix} (\Gamma_{O_K})^{-1} \Gamma_{O_K \mathbb{Z}} \\ 0  \end{matrix}\right].
\end{equation*}
Since, we will apply this operator only to sequences with support in $B_K$, we may consider

$$\forall u \in l^2(\mathbb{Z}), \operatorname{Supp}(u) \subset B_K, 
p_{O_KB_K}u :=\left[ \begin{matrix} (\Gamma_{O_K})^{-1} \Gamma_{O_K B_K} & 0 \\ 0& 0  \end{matrix} \right]u. $$
The last expression is given using the following block decomposition, if $B_K^C$ denotes the complement of $B_K$ in $\mathbb{Z}$ :
$$\left[ \begin{matrix}O_K B_K & O_K B_K^C \\ O_K^C B_K & O_K^C B_K^C \end{matrix} \right].$$

Hence, the two quantities $\Gamma_{O_K B_K}$ and $(\Gamma_{O_K})^{-1}$ have to be estimated. On the one hand, a natural estimator of the first matrix is given by $\hat{\Gamma}_{O_K B_K}$ defined as
\begin{equation*}
\left( \hat{\Gamma}_{O_KB_K}^{(N)} \right)_{ij}  =  \hat{r}^{(N)}(\left|j-i\right|), i \in O_K, j \in B_K.
\end{equation*}

On the other hand, a natural way to estimate $(\Gamma_{O_K})^{-1}$ 
could be to use $(\hat{\Gamma}^{(N)}_{O_K})$ (defined as $\left( \hat{\Gamma}_{O_K}^{(N)} \right)_{ij}  =  \hat{r}^{(N)}(\left|j-i\right|), i,j \in O_K $) and invert it.
However, it is not sure that this matrix is invertible. So, we will consider an empirical regularized version by setting
\begin{equation*}
\tilde{\Gamma}^{(N)} =  \hat{\Gamma}^{(N)}_{O_K} + \hat{\alpha}I_{O_K},
\end{equation*}
for a well chosen $\hat{\alpha}$. 

Set  
\begin{equation*} \hat{\alpha} = -\min{\hat{f}_K^{(N)}} \ind{\min{\hat{f}_K^{(N)}} \leq 0}+\frac{m}{4}\ind{\min{\hat{f}_K^{(N)}} \leq \frac{m}{4}}. \end{equation*}
so that $\left\|(\tilde{\Gamma}^{(N)}_{O_K})^{-1}\right\|_{2,op} \leq \frac{m}{4}$.
Remark that  $\tilde{\Gamma}^{(N)}$ is the Toeplitz matrix associated to the function $\tilde{f}^{(N)} =\hat{f}_K^{(N)} + \hat{\alpha}$, 
that has been tailored to ensure that $\tilde{f}^{(N)}$ is always greater than $\frac{m}{4}$, yielding the desired 
control to compute $\tilde{\Gamma}^{-1}$. 
Other regularization schemes could have been investigated. 
Nevertheless, note that adding a translation factor makes computation easier than using, for instance, a threshold on  
$\hat{f}^{(N)}_K$. Indeed, with our perturbation, we only modify the diagonal coefficients of the covariance matrix. 
\vskip .1in

Finally, we will consider the following estimator, for any $Y\in \mathcal{B}_K:=\operatorname{Span}\left( (X_i)_{i \in B_K }\right)$:
$$\hat{Y}:=\hat{P}_{O_KB_K}^{(N)}(Y) = \Phi\left(\hat{p}_{O_KB_K}^{(N)}\Phi^{-1}(Y)\right),$$
 
where the estimator $\hat{p}_{O_KB_K}^{(N)}$ of $p_{\mathbb{Z}^-B_K}$, with window $K(N)$, is defined as follows 
\begin{equation} \label{est}
\hat{p}_{O_KB_K}^{(N)} = \left(\tilde{\Gamma}^{(N)}_{O_K}\right)^{-1}\hat{\Gamma}^{(N)}_{O_KB_K}.
\end{equation}


\section{Asymptotic behavior of the empirical projection operator} \label{s:rate}


In this section, we give the rate of convergence of the estimator built previously (see Section \ref{s:frame}). We will bound uniformly the bias of prediction error for 
random variables in the close future.

First, let us give some conditions on the sequence $(K(N))_{N \in \mathbb{N}})$:
\begin{hyp}\label{a:k}
The sequence $(K(N))_{N \in \mathbb{N}}$ satisfies
\begin{itemize}
 \item $\lim K(N) \xrightarrow{N \rightarrow \infty} +\infty. $
\item  $\lim \frac{K(N)\log(K(N))}{N} \xrightarrow{N \rightarrow \infty}0 .$
\end{itemize}
\end{hyp}

Recall that the pointwise risk in $Y \in \mathbb{L}^2(\mathbb{P})$ is defined by
$$R(\hat{Y})=\mathbb{E}\left[ \left( \hat{Y}-\mathbb{E}\left[Y|(X_i)_{i <0}\right] \right)^2 \right].$$

The global risk for the window $K(N)$ is defined by taking the supremum of the pointwise risk over all random variables $Y\in \mathcal{B}_K=\operatorname{Span}\left( (X_i)_{i \in B_K }\right)$

$$\mathcal{R}_{K(N)}\left(\hat{P}^N_{O_{K}B_K}\right)= \sup_{Y \in \mathcal{B}_K,\atop  \operatorname{Var}(Y)\leq 1} R(\hat{P}^N_{O_{K}}(Y)) . $$

Notice that we could have chosen to evaluate the prediction quality only on $X_0$. Nevertheless the rate of convergence is not modified if we evaluate the prediction quality for all random variables from the close future. Indeed, the 
major part of the observations will be used for the estimation, and the conditional expectation is taken only on the most $K(N)$ recent observations. 
Our result will be then quite stronger than if we had dealt only with prediction of $X_0$.
%
%
%

To get a control on the bias of the prediction, we need some regularity assumption.
We consider Sobolev's type regularity by setting
\begin{equation*}
\forall s >1, W_s := \left\{g \in \mathbb{L}_2([0,2\pi)), g(t) = \sum_{k \in \mathbb{Z}} a_ke^{ikt}, \sum_{k \in \mathbb{Z}}k^{2s} a_k^2  < \infty \right\}.
\end{equation*}
and define
\begin{equation*}
\forall g \in W_s,  g(t) = \sum_{k \in \mathbb{Z}} a_ke^{ikt} \left\|g\right\|_{W_s}: =  \inf \left\{M, \sum_{k \in \mathbb{Z}}k^{2s}  a_k^2 \leq M \right\}.
\end{equation*}

\begin{hyp} \label{a:sobol}
There exists $s \geq 1$ such that $f^\star \in W_s$.
\end{hyp}

We can now state our results. The following lemmas may be used in other frameworks than the blind problem. More precisely, if the blind prediction problem is very specific, 
the control of the loss between prediction with finite and infinite past is more classical, and the following lemmas may be applied for that kind of questions.
The case where independent samples are available may also be tackled with the last estimators, using rates of convergences given in operator norms.

The bias is given by the following lemma
\begin{lem} \label{l:bias}
For $N$ large enough, the following upper bound holds,
\begin{equation*}
\left\| p_{O_KB_k} - p_{\mathbb{Z}^-B_K}  \right\|_{\Gamma} \leq C_2\frac{1}{K(N)^{\frac{2s-1}{2}}}, 
\end{equation*}
where $C_2 = \left\|\frac{1}{f^\star}\right\|_{W_{2s}} m'(1+\frac{m'}{m})$.
\end{lem} 

In the last lemma, we assume regularity in terms of Sobolev's classes. Nevertheless, the proof may be written with some other kind of regularity. 
The proof is given in appendix, and is essentially based on Proposition \ref{p:schur}. This last proposition provides the Schur block inversion of the projection operator.

The control for the variance is given in the following lemma:
\begin{lem}\label{l:var}
\begin{equation*}
\int_0^\infty \mathbb{P}\left(\left\| \hat{p}_{O_KB_K}^N - p_{O_KB_K} \right\|_{\Gamma}^4 > t \right) \dd t 
\leq C_0^4K(N)^4 (\frac{\log(K(N))}{N})^2 + o(K(N)^4 (\frac{\log(K(N))}{N})^2), 
\end{equation*}
where $C_0= 4m'(\frac{6m'}{m^2} +\frac{4}{m}+2)$ 
\end{lem}

Again, we choose this concentration formulation to deal with the dependency of the blind prediction problem, but this result gives immediately a control of the variance of 
the estimator whenever independent samples are observed (one for the estimation, and another one for the prediction).

The proof of this lemma is given in Section \ref{s:proof_lemma}. It is based on a concentration inequality of the estimators
$\hat{r}^{(N)}_p$ (see Comte \cite{comte}).  

Integrating this rate of convergence over the blind data, we get our main theorem.

\begin{thm}\label{t:main}
Under Assumptions \ref{a:fbornee}, \ref{a:k} and \ref{a:sobol}, for $N$ large enough, the empirical estimator satisfies
\begin{equation*}
\sqrt{\mathcal{R} (\hat{P}^{(N)}_{O_KB_K} ) }\leq  C_1\frac{K(N)^2\sqrt{\log(K(N))}}{\sqrt{N}} +C_2\frac{1}{K(N)^{\frac{2s-1}{2}}}, 
\end{equation*}
where $C_1$ and $C_2$ are given in Appendix.
\end{thm} 

Again, the proof of this result is given in Section \ref{s:proof_main_thm}. It is quite technical. The main difficulty is induced by the blindness. Indeed, in this step, we have to deal with the dependency between the data and the empirical projector.

Obviously, the best rate of convergence is obtained by balancing the variance and the bias and finding the best window $K(N)$. 
Indeed, the variance increases with $K(N)$ while the bias decreases.
Define $\hat{P}^{(N)}_\star$ as the projector $\hat{P}^{(N)}_{K^\star(N)}$ associated to the sequence $K^\star(N)$ that minimizes the bound in the last theorem. 
We get:
\begin{cor}[Rate of convergence of the prediction estimator]
\label{tmain}
Under Assumptions $1.1$ and $2.1$, for $N$ large enough and choosing $K(N) = \Big{\lfloor} (\frac{N}{\log N})^{\frac{1}{2(2s+3)}} \Big{\rfloor}$, we get
 \begin{equation} \label{rate}
\sqrt{\mathcal{R}(\hat{P}^{(N)}_\star)} \leq O\left( \left(\frac{\log N}{N}\right)^\frac{2s-1}{2(2s+3)}\right).
\end{equation} 
\end{cor}

Notice that, in real life issues, it would be more natural to balance the risk given in Theorem \ref{t:main}, with the macroscopic term of variance given by
$$\mathbb{E}\big[Y-\mathbb{E}\left[Y|(X_i)_{i<0}\right] \big].$$
%
This leads to a much greater $K(N)$.
Nevertheless, Corollary \ref{tmain} has a theoretical interest. Indeed, it recovers the classical semi-parametric rate of convergence, and provides a way to get away from dependency. 
Notice that, the estimation rate increases with the regularity $s$ of the spectral density $f^\star$. 
More precisely, if $s\rightarrow \infty$, we obtain $(\frac{\log N}{N})^{\frac{1}{2}}$. This is, up to the $\log$-term, the optimal speed.
As a matter of fact, in this case, estimating the first coefficients of the covariance matrix is enough. Hence, the bias is very small. 
Proving a lower bound on the mean error (that could lead to a minimax result), is a difficult task, since the 
tools used to design the estimator are far from the usual estimation methods. \vskip .1in



 \section{Projection onto finite observations with known covariance}\label{section_schur}

 We aim at providing an exact expression for the projection operator.  
For this, we generalize the expression given by Bondon (\cite{bondon}, \cite{bondon2}) for a projector onto infinite past. 
%
%
Recall that, for any $A \subset \mathbb{Z}$, and if $A^C$ denotes the complement of $A$ in $\mathbb{Z}$, the projector $p_A$ may be written blockwise (see for instance \cite{Ste}) as:
$$p_A = \left[\begin{matrix} Id_A & \Gamma_A^{-1} \Gamma_{AA^C}  \\0 & 0 \end{matrix}\right] .$$ 
Denote also $\Lambda := \Gamma^{-1} = T(\frac{1}{f^\star})$ the inverse of the covariance operator, 
the following proposition provides an alternative expression of any projection operators.  

\begin{prop}\label{p:schur}
One has
\begin{eqnarray*}
p_A 
 & = &
=\left[\begin{matrix} \operatorname{Id}_A &  - \Lambda_{AA^C}\Lambda_{A^C}^{-1}  \\0 & 0 \end{matrix}\right] 
\end{eqnarray*}

Furthermore, the prediction error verifies $$\mathbb{E}\left[\left(P_A Y-Y\right)^2\right]=u^T\Lambda_{M}^{-1}u,$$
 where $Y = \Phi(u) = u^T X.$ 
\end{prop}
\noindent The proof of this proposition is given in Appendix. We point out that this proposition is helpful for the computation of the bias.
 Indeed, it gives a way to calculate the norm of the difference between two inverses operators.  

\section{Appendix} \label{s:append}

\subsection{Proof of Proposition~\ref{p:schur}}
\begin{proof}
For the proof of Proposition \ref{p:schur}, let us choose 
$$A \subset \mathbb{Z},$$
and denote the complement of $A$ in $\mathbb{Z}$ by
$$M := A^C $$

First of all, $\Lambda = \Gamma^{-1}$ is a Toeplitz operator over $H$ with eigenvalues in $[\frac{1}{m'} ;\frac{1}{m}]$. $\Lambda_M$ 
may be inverted as a principal minor of $\Lambda$.
Let us define the Schur complement of $\Lambda$ on sequences with support in $M$ : $S = \Lambda_A - \Lambda_{AM}\Lambda_{M}^{-1}\Lambda_{MA}$. 
The next lemma provides an expression of $\Gamma_{A}^{-1}$ (see for instance \cite{schurbook}).
\begin{lem}\label{l:schur}
\begin{eqnarray*}
\Gamma_{A}^{-1} &= &S \\ & = &\Lambda_A - \Lambda_{AM}\Lambda_{M}^{-1}\Lambda_{MA}.
\end{eqnarray*}
\end{lem}

\begin{proof} of Lemma \ref{l:schur}

One can check 
\begin{eqnarray*}
& &\left[\begin{matrix}\Lambda_A &\Lambda_{AM} \\ \Lambda_{MA} & \Lambda_M \end{matrix} \right] \left[\begin{matrix}S^{-1}&- S^{-1}\Lambda_{AM}\Lambda_M^{-1} \\ - \Lambda_M^{-1}\Lambda_{MA}S^{-1} & \Lambda_M^{-1} + \Lambda_M^{-1}\Lambda_{MA}S^{-1}\Lambda_{AM}\Lambda_M^{-1} \end{matrix}\right] \\
& = & \left[\begin{matrix}\Lambda_AS^{-1} -\Lambda_{AM}\Lambda_M^{-1}\Lambda_{MA}S^{-1} & -\Lambda_AS^{-1}\Lambda_{AM}\Lambda_M^{-1} + \Lambda_{AM}(\Lambda_M^{-1} + \Lambda_M^{-1}\Lambda_{MA}S^{-1}\Lambda_{AM}\Lambda_M^{-1}) \\ \Lambda_{MA}S^{-1} - \Lambda_M\Lambda_M^{-1}\Lambda_{MA}S^{-1} & - \Lambda_{MA}S^{-1}\Lambda_{AM}\Lambda_M^{-1} + \Lambda_M (\Lambda_M^{-1} + \Lambda_M^{-1}\Lambda_{MA}S^{-1}\Lambda_{AM}\Lambda_M^{-1} )  \end{matrix}\right]
\\ & = & \left[\begin{matrix}SS^{-1} & (\Lambda_{AM} \Lambda_M^{-1}\Lambda_{MA}S^{-1}+I_A-\Lambda_AS^{-1})\Lambda_{AM}\Lambda_M^{-1} \\ \Lambda_{MA}S^{-1}- \Lambda_{MA}S^{-1} & - \Lambda_{MA}  S^{-1}\Lambda_{AM}\Lambda_M^{-1}  + I_M +\Lambda_{MA}S^{-1}\Lambda_{AM}\Lambda_M^{-1}   \end{matrix}\right]
\\ &=&\left[\begin{matrix}I_A& 0 \\ 0 & I_M   \end{matrix}\right].
\end{eqnarray*}

Since the matrix are symmetric, we can transpose the last equality. We obtain that
\begin{eqnarray*}
\left[\begin{matrix}S^{-1}&- S^{-1}\Lambda_{AM}\Lambda_M^{-1} \\ - \Lambda_M^{-1}\Lambda_{MA}S^{-1} & \Lambda_M^{-1} + \Lambda_M^{-1}\Lambda_{MA}S^{-1}\Lambda_{AM}\Lambda_M^{-1} \end{matrix}\right] & = & \Lambda^{-1}
\\ &= &\Gamma.
\end{eqnarray*}
So that $\Gamma_A = S^{-1} $.
\end{proof}
We now compute the projection operator:
\begin{eqnarray*}
p_A &= & 
\left[\begin{matrix} Id_A & \Gamma_A^{-1}\Gamma_{AM}  \\0 & 0 \end{matrix}\right] \\
&  = & \left[\begin{matrix} Id_A & S \Gamma_{AM}  \\0 & 0 \end{matrix}\right]\\
&  = & \left[\begin{matrix} Id_A &  (\Lambda_A - \Lambda_{AM}\Lambda_{M}^{-1}\Lambda_{MA})\Gamma_{AM}  \\0 & 0 \end{matrix}\right]\\
& = & \left[\begin{matrix} Id_A &  \Lambda_A \Gamma_{AM} - \Lambda_{AM}\Lambda_{M}^{-1}(Id_M - \Lambda_M \Gamma_M)  \\0 & 0 \end{matrix}\right]\\
& = & \left[\begin{matrix} Id_A &  \Lambda_A \Gamma_{AM} - \Lambda_{AM}\Lambda_{M}^{-1} + \Lambda_{AM} \Gamma_M  \\0 & 0 \end{matrix}\right]\\
& = & \left[\begin{matrix} Id_A &  - \Lambda_{AM}\Lambda_{M}^{-1}  \\0 & 0 \end{matrix}\right].
\end{eqnarray*}
Where we have used $\Lambda  \Gamma = Id$ in the last two lines.

Now consider $Q$ the quadratic error operator. It is defined as
$$\forall u \in l^2(\mathbb{Z}), u^TQu := \left\|(p_Au-u)^2 \right\|_{\Gamma} = \mathbb{E}\left[(\Phi(u)- P_A \Phi(u) )^2\right].   $$  

This operator $Q$ can be obtained by a direct computation (writing the product right above), but it is easier to use the expression of the variance of a projector in the Gaussian case given for instance by \cite{Ste}.
\begin{equation*}
Q = \Gamma_M-\Gamma_{MA}\Gamma_{A}^{-1} \Gamma_{AM}
\end{equation*}
Again, notice that $Q$ is the Schur complement of $\Gamma$ on sequences with support in $A$, and thanks to Lemma~\ref{l:schur} applied to $\Lambda$ instead of $\Gamma$,
 we get
\begin{equation*}
Q = \Lambda_M^{-1}.
\end{equation*}
This ends the proof of Proposition \ref{p:schur}.
\end{proof}

\subsection{Proof of Theorem \ref{t:main}}
\label{s:proof_main_thm}

\begin{proof} of Theorem \ref{t:main}

Recall that we aim at providing a bound on  $\sqrt{\mathcal{R}(\hat{P}^{(N)}_{O_KB_K})}$.

Notice first that we have
$$\sqrt{\mathcal{R}(\hat{P}^{(N)}_{O_KB_K})} \leq  \sqrt{\sup_{Y \in \mathcal{B}_K,\atop  \operatorname{Var}(Y)\leq 1}
\mathbb{E}\left[ (\hat{P}_{O_KB_K}^{(N)}(Y) - P_{O_KB_K}(Y))^2\right] } + \sqrt{\mathcal{R}(P_{O_KB_K}))}  . $$

Using Lemma \ref{l:bias} for a sequence $(K(N))_{N \in \mathbb{N}}$ and a centered random variable $Y \in \operatorname{Span}\left((X_i)_{i \in B_K}\right)$ such that $\mathbb{E}\left[Y^2\right] = 1$, we have 
\begin{eqnarray*}
\sqrt{\mathcal{R}(P_{O_KB_K}))} & \leq  & \left\| p_{O_KB_K} - p_{\mathbb{Z}^-B_K}  \right\|_{\Gamma} \sqrt{\mathbb{E}\left[Y^2\right] }
\\ & \leq & C_2\frac{1}{K(N)^{\frac{2s-1}{2}}} .
\end{eqnarray*}

For the variance, we first notice that $Y = \Phi(u) = u^T\mathbf{X}$,
\begin{equation*} 1 = \mathbb{E}\left[Y^2\right]  = u^T\Gamma_{B_K}u \geq m u^Tu = m\sum_{i=0}^{K(N)-1}u_i^2,\end{equation*}

Denote $A=\hat{p}_{O_KB_K}^{(N)} - p_{O_KB_K}$. We can write, by applying twice Cauchy-Schwarz's inequality,
\begin{eqnarray*}
\mathbb{E}\left[\Big(\hat{P}_{O_KB_K}^{(N)}Y - P_{O_KB_K}Y\Big)^2\right] &  = & \int_\omega \Big(  \sum_{ i = -K(N)}^{-1} \sum_{j = 0}^{K(N)-1} A_{ij}(\omega)u_jX_i(\omega)\Big)^2\dd \mathbb{P}(\omega)
\\ & \leq & \int_\omega \sum_{i = -K(N)}^{-1} (\sum_{j = 0}^{K(N)-1} A_{ij}(\omega)u_j)^2 \sum_{i = -K(N)}^{-1} X_i^2(\omega)\dd \mathbb{P}(\omega)
\\ & \leq & \int_\omega \sum_{ i = -K(N)}^{-1} \sum_{j = 0}^{K(N)-1}  A_{ij}^2(\omega) \sum_{j = 0}^{K(N)-1} u_j^2 \sum_{i = -K(N)}^{-1} X_i^2(\omega)\dd 
\mathbb{P}(\omega),
\end{eqnarray*}
So that,
\begin{eqnarray*}
 \mathbb{E}\left[\Big(\hat{P}_{O_KB_K}^{(N)}Y - P_{O_KB_K}Y\Big)^2\right]  & \leq & \int_\omega   \sum_{ i = -K(N)}^{-1}\sum_{j = 0}^{K(N)-1} A_{ij}^2(\omega)  \frac{1}{m} \sum_{i = n_0+1}^{K(N)+n_0} X_i^2\dd \mathbb{P}(\omega).
\end{eqnarray*}
Using the following equivalence between two norms for finite matrices with size $(n,m)$ (see for instance \cite{matrixhandbook}),

\begin{equation*}
\sqrt{\sum_{i = 1}^{n} \sum_{ j = 1}^{m} A_{ij}^2} \leq  \sqrt{n} \left\| A \right\|_{2,op},
\end{equation*}
we obtain

\begin{eqnarray*}
\mathbb{E}\left[\Big(\hat{P}_{O_KB_K}^{(N)}Y - P_{O_KB_K}Y\Big)^2\right]  & \leq & \frac{K(N)}{m} \int_\omega  \left\| A (\omega)\right\|_{2,op}^2  \sum_{i = n_0+1}^{K(N)+n_0} X_i^2(\omega)\dd \mathbb{P}(\omega).
\end{eqnarray*}

%
%
Further, 

\begin{eqnarray*}
\mathbb{E}\left[\Big(\hat{P}_{O_KB_K}^{(N)}Y - P_{O_KB_K}Y\Big)^2\right] 
 & \leq & \frac{K(N)}{m} \int_\omega  \left\| A(\omega) \right\|_{2,op}^2  \sum_{j = n_0+1}^{K(N)+n_0} X_j^2(\omega)\dd \mathbb{P}(\omega) 
\\ & \leq & \frac{K(N)}{m} \sqrt{\int_\omega  \left\| A (\omega)\right\|_{2,op}^4\dd \mathbb{P}(\omega)} \sqrt{\int_\omega\left(\sum_{j = n_0+1}^{K(N)+n_0} X_j^2(\omega)\right)^2\dd \mathbb{P}(\omega)}
\\ & \leq & \frac{K(N)}{m} \sqrt{\int_{\mathbb{R}^+} \mathbb{P}\left( \left\| A \right\|_{2,op}^4>t \right)\dd t} \sqrt{K(N)^2\int_\omega\left( X_j^4\right)\dd \mathbb{P}(\omega)},
\end{eqnarray*}
We have used here again Cauchy-Schwarz's inequality and the fact that, for all nonnegative random variable $Y$,  
\begin{equation*}
\mathbb{E}\left[Y\right] = \int_{\mathbb{R}^+} \mathbb{P}\left( Y>t\right) \dd t.
\end{equation*}

Since $X_0$ is Gaussian, its moment of order four $r_4$ is finite. Then Lemma \ref{l:var} yields that, for $N$ large enough,

\begin{equation*}
\mathbb{E}\left[\Big(\hat{P}_{O_KB_K}^{(N)}Y - P_{O_KB_K}Y\Big)^2\right]  \leq \frac{C_0^2\sqrt{r_4}K(N)^4\log(K(N)))}{mN}.
\end{equation*}
So that,
\begin{equation*}
\sqrt{\sup_{Y \in \mathcal{B}_K,\atop  \operatorname{Var}(Y)\leq 1}
\mathbb{E}\left[ (\hat{P}_{O_K}^{(N)}(Y) - P_{O_K}(Y))^2\right] } \leq \frac{C_1K(N)^2\sqrt{\log(K(N))}}{\sqrt{N}},
\end{equation*}
with $C_1 = \frac{C_0\sqrt[4]{r_4}}{\sqrt{m}}$. This ends the proof of the theorem.
\end{proof}

\subsection{Proofs of concentration and regularity lemmas }
\label{s:proof_lemma}

First, we compute the bias and prove Lemma~\ref{l:bias} :
\begin{proof}{of Lemma \ref{l:bias}}

Recall that we aim to obtain a bound on $\left\| p_{O_KB_K} - p_{\mathbb{Z}^-B_K} \right\|_{\Gamma}$.
Using Proposition~\ref{p:schur}, we can write 
\begin{eqnarray*}
\left\| p_{O_KB_K} - p_{\mathbb{Z}^-B_K} \right\|_{\Gamma} &\leq& \left\| p_{O_K\mathbb{Z}^+} - p_{\mathbb{Z}^-\mathbb{Z}^+} \right\|_{\Gamma} 
\\ & \leq & \left\| \left[\begin{matrix} (\Gamma_{O_K})^{-1}\Gamma_{O_K\mathbb{Z}^+} \\ 0 \end{matrix}\right] - \left[\begin{matrix} - \Lambda_{O_K\mathbb{Z}^+}(\Lambda_{\mathbb{Z}^+})^{-1} \\
 -\Lambda_{M_K^-\mathbb{Z}^+}(\Lambda_{\mathbb{Z}^+})^{-1} \end{matrix}\right] \right\|_{\Gamma}.
\end{eqnarray*}
So that, using the norms equivalence,
\begin{eqnarray*}
\left\| p_{O_KB_K} - p_{\mathbb{Z}^-B_K} \right\|_{\Gamma}& \leq &  \frac{m'}{m} \left\| \left[\begin{matrix} (\Gamma_{O_K})^{-1}\Gamma_{O_K\mathbb{Z}^+} \\ 0 \end{matrix}\right] - \left[\begin{matrix} - \Lambda_{O_K\mathbb{Z}^+}(\Lambda_{\mathbb{Z}^+})^{-1} \\
 -\Lambda_{M_K^-\mathbb{Z}^+}(\Lambda_{\mathbb{Z}^+})^{-1} \end{matrix}\right] \right\|_{2,op} \\
& \leq & \frac{m'}{m} \left\| \left[ \begin{matrix} (\Gamma_{O_K})^{-1}\Gamma_{O_K\mathbb{Z}^+} + \Lambda_{O_K\mathbb{Z}^+}(\Lambda_{\mathbb{Z}^+})^{-1} \\ \Lambda_{M_K^-\mathbb{Z}^+}(\Lambda_{\mathbb{Z}^+})^{-1} \end{matrix}\right]  \right\|_{2,op} \\
 & \leq &  \frac{m'}{m}\left\| \left[ \begin{matrix} (\Gamma_{O_K})^{-1}\Gamma_{O_K\mathbb{Z}^+}\Lambda_{\mathbb{Z}^+} + \Lambda_{O_K\mathbb{Z}^+}\\ \Lambda_{M_K^-\mathbb{Z}^+}\end{matrix}\right]  \right\|_{2,op} \left\|   (\Lambda_{\mathbb{Z}^+})^{-1} \right\|_{2,op}\\
 & \leq &  \frac{m'}{m}\left\| (\Lambda_{\mathbb{Z}^+})^{-1} \right\|_{2,op} \left( \left\|(\Gamma_{O_K})^{-1}\Gamma_{O_K\mathbb{Z}^+}\Lambda_{\mathbb{Z}^+}  + \Lambda_{O_K\mathbb{Z}^+} \right\|_{2,op}  + \left\| \Lambda_{M_K^-\mathbb{Z}^+} \right\|_{2,op} \right).
\end{eqnarray*}
The last step follows from the inequality:
\begin{equation*}\left\| \begin{matrix}A \\ B \end{matrix} \right\|_{2,op} \leq \left\| \begin{matrix}A \\ 0 \end{matrix} \right\|_{2,op} + \left\|  \begin{matrix}0 \\ B \end{matrix}  \right\|_{2,op} = \left\| A  \right\|_{2,op} + \left\|  B   \right\|_{2,op}.
\end{equation*}
But, since $\Lambda= \Gamma^{-1}$, 
\begin{equation*}
\Gamma_{O_K\mathbb{Z}^+}\Lambda_{\mathbb{Z}^+} + \Gamma_{O_K}\Lambda_{O_K\mathbb{Z}^+} = - \Gamma_{O_KM_K^-}\Lambda_{M_K^-\mathbb{Z}^+}.
\end{equation*}
So, we obtain, 
\small{\begin{eqnarray*}
\left\| p_{O_KB_K} - p_{\mathbb{Z}^-B_K} \right\|_{\Gamma}
& \leq & \frac{m'}{m}\left\| (\Lambda_{\mathbb{Z}^+})^{-1} \right\|_{2,op} \left( \left\|(\Gamma_{O_K})^{-1}\left(- \Gamma_{O_KM_K^-}\Lambda_{M_K^-\mathbb{Z}^+}\right) \right\|_{2,op}  + \left\| \Lambda_{M_K^-\mathbb{Z}^+} \right\|_{2,op} \right)\\
& \leq & \frac{m'}{m} \left\| (\Lambda_{\mathbb{Z}^+})^{-1} \right\|_{2,op}  \left( \left\|(\Gamma_{O_K})^{-1} \right\|_{2,op}  \left\| - \Gamma_{O_KM_K^-}\Lambda_{M_K^-\mathbb{Z}^+} \right\|_{2,op} + \left\| \Lambda_{M_K^-\mathbb{Z}^+} \right\|_{2,op} \right)\\
& \leq &  \frac{m'}{m}\left\| (\Lambda_{\mathbb{Z}^+})^{-1} \right\|_{2,op} \left( \left\|(\Gamma_{O_K})^{-1} \right\|_{2,op} \left\| \Gamma_{O_K M_K^-} \right\|_{2,op} \left\| \Lambda_{M_K^-\mathbb{Z}^+} \right\|_{2,op} + \left\| \Lambda_{M_K^-\mathbb{Z}^+} \right\|_{2,op} \right)  \\
& \leq &  \frac{m'}{m} \left\| (\Lambda_{\mathbb{Z}^+})^{-1} \right\|_{2,op} \left(  \left\|(\Gamma_{O_K})^{-1} \right\|_{2,op} \left\| \Gamma_{O_K M_K^-} \right\|_{2,op} + 1 \right) \left\| \Lambda_{M_K^-\mathbb{Z}^+} \right\|_{2,op} .
\end{eqnarray*}}
But, we have,

 \begin{equation*}\left\| (\Lambda_{\mathbb{Z}^+})^{-1} \right\|_{2,op} \leq m',   \end{equation*} as the inverse of a principal minor of $\Lambda$.

 \begin{equation*}\left\|(\Gamma_{O_K})^{-1} \right\|_{2,op} \leq \frac{1}{m},  \end{equation*} since it is the inverse of a principal minor of $\Gamma$.

 \begin{equation*} \left\| \Gamma_{O_KM_K^-} \right\|_{2,op} \leq m'  ,\end{equation*} as an extracted operator of $\Gamma$.

Thus, we get
\begin{equation*}
\left\| p_{O_KB_K} - p_{\mathbb{Z}^-B_K} \right\|_{\Gamma}  \leq C_4 \left\| \Lambda_{M_K^-\mathbb{Z}^+} \right\|_{2,op}, 
\end{equation*}

where $C_4 = \frac{m'^2}{m}(1+\frac{m'}{m})$.
Since $f^\star \in H_s$ (Assumption $2.1$), and $f^\star \geq m > 0$, we have also $\frac{1}{f^\star} \in H_s$. If we denote $p(k) = \Lambda_{i,i+k}$ the Fourier coefficient of $\frac{1}{f^\star}$, we get 
\begin{eqnarray*}
\left\| \Lambda_{M_K^-\mathbb{Z}^+} \right\|_{2,op}  & \leq & \left\|\Lambda_{M_K^-\mathbb{Z}^+} \right\|_2 \\
& \leq & \sqrt{\sum_{i \leq - K(N) ; 0 \leq j} p(j-i)^2}\\
& \leq & \sqrt{\sum_{i = K(N)}^\infty \sum_{j= i}^\infty p(j)^2} \\
& \leq & \sqrt{\sum_{i = K(N)}^\infty \frac{\left\|\frac{1}{f^\star}\right\|_{W_{s}}}{i^{2s}}}\\
& \leq &\sqrt{\left\|\frac{1}{f^\star}\right\|_{H_{s}}\frac{1}{K(N)^{s-1}}}.
\end{eqnarray*}
\end{proof}

So that the lemma is proved and the bias is given by

\begin{equation*}
\left\| p_{O_KB_K} - p_{\mathbb{Z}^-B_K} \right\|_{\Gamma} \leq C_4\sqrt{\left\|\frac{1}{f^\star}\right\|_{W_{s}}}\frac{1}{K(N)^{\frac{2s-1}{2}}}.
\end{equation*}
Actually, the rate of convergence for the bias is given by the regularity of the spectral density, since it depends on the coefficients far away from the principal diagonal.\vskip .1in

Now, we prove Lemma \ref{l:var}, which achieves the proof of the theorem.

\begin{proof}{ of Lemma \ref{l:var}}

Recall that $A=\hat{p}_{O_KB_K}^{(N)} - p_{O_KB_K}$. We aim at proving that 

$$\int_0^\infty \mathbb{P}\left(\left\| A \right\|_{\Gamma}^4 > t \right) \dd t 
\leq C_0^4K(N)^4 (\frac{\log(K(N))}{N})^2 + o(K(N)^4 (\frac{\log(K(N))}{N})^2).$$

First, 
\begin{eqnarray*}
\left\| A \right\|_{2,op} & = & \left\| (\tilde{\Gamma}_{O_K})^{-1}\hat{\Gamma}_{O_KB_K} - (\Gamma_{O_K})^{-1}\Gamma_{O_KB_K} \right\|_{2,op} 
\\ &\leq  & \left\| \Gamma_{O_KB_K}\right\|_{2,op} \left\| (\tilde{\Gamma}_{O_K})^{-1}-(\Gamma_{O_K})^{-1} \right\|_{2,op} + \left\|(\tilde{\Gamma}_{O_K})^{-1}\right\|_{2,op} \left\| \hat{\Gamma}_{O_KB_K} - \Gamma_{O_KB_K} \right\|_{2,op}
\\& \leq & \left\| \Gamma_{O_KB_K} \right\|_{2,op} \left\| (\tilde{\Gamma}_{O_K})^{-1} \right\|_{2,op} \left\|(\Gamma_{O_K})^{-1} \right\|_{2,op} \left\| \tilde{\Gamma}_{O_K} - \Gamma_{O_K} \right\|_{2,op}
\\ & & + \left\|(\tilde{\Gamma}_{O_K})^{-1}\right\|_{2,op} \left\|\hat{\Gamma}_{O_KB_K}- \Gamma_{O_KB_K} \right\|_{2,op} .
\end{eqnarray*}

But, we have,

\begin{equation*}\left\| \Gamma_{O_KB_K}\right\|_{2,op} \leq m',   \end{equation*} as an extracted operator of $\Gamma$.

 \begin{equation*}\left\|(\Gamma_{O_K})^{-1} \right\|_{2,op} \leq \frac{1}{m},  \end{equation*} as the inverse of a principal minor of $\Gamma$.

\begin{equation*} \left\| (\tilde{\Gamma}_{O_K})^{-1} \right\|_{2,op} \leq \frac{4}{m}  , \end{equation*} thanks to the regularization.
Furthermore,
 \begin{equation*}\left\| \tilde{\Gamma}_{O_K} - \Gamma_{O_K} \right\|_{2,op} \leq K(N)\sup_{p\leq 2 K(N)} \left\{\left|\hat{r}_N(p)-r(p)\right|\right\}+ 
 \left|\hat{\alpha}\right|. \end{equation*}

So the regularization also gives
 \begin{equation*}\left\|\hat{\Gamma}_{O_KB_K}- \Gamma_{O_KB_K} \right\|_{2,op}  \leq K(N)\left(\sup_{p\leq 2 K(N)} \left\{ \left|\hat{r}_N(p)-r(p)\right| \right\}\right) .\end{equation*}

 \begin{equation*} \left| \hat{\alpha}\right| =  \left|-\min{\hat{f}_K^N} \ind{\min{\hat{f}_K^N} \leq 0}+\frac{m}{4}\ind{\min{\hat{f}_K^N} \leq \frac{m}{4}}   \right|.
 \end{equation*}

So, \begin{equation*} \left| \hat{\alpha}\right| \leq  (2K(N)+1)\sup_{p\leq 2K(N)} \left\{ \hat{r}_N(p)-r(p)\right\}+\frac{m}{4}\ind{\min{\hat{f}_K^N} \leq \frac{m}{4}}.
 \end{equation*}

For the last inequality, we used the following lemma, proved in the next section.
\begin{lem}\label{lemspec}
The empirical spectral density is such that, for $N$ large enough
\begin{equation*}
\left\| \hat{f}_{K(N)}^N - f^\star    \right\|_\infty \leq  (2K(N)+1)\sup_{p\leq 2K(N)} \left\{ \hat{r}_N(p)-r(p)\right\} +\frac{m}{4}.
\end{equation*}
\end{lem}
This implies
\begin{equation*}
\left| \min{\hat{f}_K^N} \ind{\min{\hat{f}_K^N} \leq 0} \right| \leq (2K(N)+1)\sup_{p\leq 2K(N)} \left\{ \hat{r}_N(p)-r(p)\right\}.
\end{equation*}

So, we obtain,
\begin{eqnarray*}  
\left\| A \right\|_{2,op} & \leq & \frac{4m'}{m^2}\left(K(N)\sup_{p\leq 2 K(N)} \left\{\left|\hat{r}_N(p)-r(p)\right|\right\}+ \left|\hat{\alpha}\right|\right) + \frac{4}{m}K(N)\left(\sup_{p\leq 2 K(N)} \left\{\left|\hat{r}_N(p)-r(p)\right|\right\}\right)
\\ & \leq & \left( \frac{6m'}{m^2} +\frac{4}{m} + 2 + \frac{1}{K(N)} \right)K(N)\left(\sup_{p\leq 2 K(N)} \left\{\left|\hat{r}_N(p)-r(p)\right|\right\}\right) + \frac{m'}{m}\ind{\min{\hat{f}_K^N \leq \frac{m}{4}}} .\end{eqnarray*}

We will use here some other technical lemmas. Their proofs are also postponed to the last section. The first one gives an uniform concentration result on the estimator $\hat{r}_N(p)$: 
\begin{lem}\label{lemconc}
Assume that Assumption \ref{a:k} holds. Then, there exists $N_0$ such that, for all $N \geq N_0$, and $x \geq 0$,
\begin{equation*}
\forall p \leq 2K(N), \left|\hat{r}_N(p)-r(p)\right|> 4m' \left(\sqrt{\frac{(\log(K(N))+x)}{N}} + \frac{x}{N}\right),
\end{equation*}
with probability at least $1-e^{-x}$
\end{lem}

For ease of notations, we set $C_0 = 4m'\left( \frac{6m'}{m^2} +\frac{4}{m} + 2\right)$ and $C_3 = \frac{m'}{m}.$
For the computation of the mean, the interval $[0,+\infty[$ will be divided into three parts, where only the first contribution is significant, thanks to the exponential concentration. We will prove that the two other parts are negligible.

We obtain, for all $x\geq 0$ 
\begin{equation*}
\left\| A \right\|_{2,op} \leq  (C_0 + o(1)) K(N)\left(\sqrt{\frac{\log(K(N)+x}{N}}+\frac{x}{N}\right) + C_3\ind{\min{\hat{f}_K^N \leq \frac{m}{4}}},
\end{equation*}
with probability at least $1-e^{-x}$
$\\$
$\\$
Set $t_1 = \left(C_0K(N)\sqrt{\frac{\log(K(N))}{N}}\right)^4$.

For $t \in [0,t_1]$, we use the inequality \begin{equation*} \mathbb{P}\left(\left\| A \right\|_{2,op}^4>t\right) \leq 1 .\end{equation*}

We obtain the first contribution to the integral. This is also the non negligible part.
\begin{equation*}
\int_0^{t_1} \mathbb{P}\left(\left\| A \right\|_{2,op}^4>t\right)\dd t = \left(C_0K(N)\sqrt{\frac{\log(K(N))}{N}}\right)^4.
\end{equation*}

$\\$
$\\$
$\\$

Now, set $t_2 = \left(C_0K(N)\sqrt{\frac{\log(K(N))+N}{N}}+C_3\right)^4$.

For $t \in [t_1,t_2]$, we use 
\small{\begin{equation*}
\mathbb{P}\left(\left\| A \right\|_{2,op}^4>\sup \left(C_0^4K(N)^4\left(\frac{\log(K(N))+x}{N}\right)^2,C_0^4K(N)^4\left(\frac{x}{N}\right)^4\right)\right) \leq e^{-x} + \mathbb{P}\left(\min{\hat{f}_K^N \leq \frac{m}{4}}\right).
\end{equation*}}

Notice that the last lemma provides 
\begin{equation*}
\mathbb{P}\left(2K(N)\sup_{p\leq 2 K(N)} \left\{\left|\hat{r}_N(p)-r(p)\right|\right\} > \frac{m}{2}\right) \leq e^{-\frac{Nm^2}{(64K(N)m')^2}}).
\end{equation*}
Indeed, set $x_0(N) = \frac{Nm^2}{(64K(N)m')^2}$.

One can compute that with probability at least $1-e^{-x_0(N)}$,
\begin{eqnarray*}
\sup_{p\leq 2 K(N)} \left\{\left|\hat{r}_N(p)-r(p)\right|\right\} &\leq &4m'\left(\sqrt{\frac{\log(K(N)) +x_0(N) }{N}} + \frac{x_0(N)}{N} \right)
\\ & \leq & 4m'\left(  \sqrt{\frac{\log(K(N))}{N}+ \frac{m^2}{(64K(N)m')^2}} + \frac{m^2}{(64K(N)m')^2}  \right)
\\ & \leq & 4m'\left(  \sqrt{\frac{\log(K(N))}{N}} + \sqrt{\frac{m^2}{(64K(N)m')^2}} + \frac{m^2}{(64K(N)m')^2}  \right)
\\ & \leq & 4m'\left(  \sqrt{\frac{\log(K(N))}{N}} + \frac{m}{(64K(N)m')} + \frac{m^2}{(64K(N)m')^2}  \right)
\\ & \leq & \frac{m}{8K(N)},
\end{eqnarray*}
for $N$ large enough.
Hence, 
\begin{equation*}
\mathbb{P}\left(\min{\hat{f}^N_{K} \leq \frac{m}{4}}\right) \leq e^{-\frac{Nm^2}{(64K(N)m')^2}}.
\end{equation*}
So, we have 
\begin{equation*}
\mathbb{P}\left(\left\| A \right\|_{2,op}^4> \max \left(C_0^4K(N)^4\left(\frac{\log(K(N))+x}{N}\right)^2,C_0^4K(N)^4\left(\frac{x}{N}\right)^4\right) \right) \leq e^{-x}+ e^{-\frac{Nm^2}{(64K(N)m')^2}}.
\end{equation*}

Finally, the following lemma (the proof is again postponed in Appendix) will be useful to transform a probability inequality into an $\mathbb{L}^2$ inequality.
\begin{lem}\label{lemfun}
Let $X$ be a nonnegative random variable such that there exists two one to one maps $f_1$ and $f_2$ and a $C>0$  with 
\begin{equation*}
\forall x \geq 0,  \mathbb{P}\left( X > \sup(f_1(x),f_2(x))\right) \leq e^{-x} + C,
\end{equation*}
then \begin{equation*}
\mathbb{P}\left( X > t \right) \leq e^{-f_1^{-1}(t)}+e^{-f_2^{-1}(t)} + C.
\end{equation*}
\end{lem}

So, thanks to lemma \ref{lemfun}, we have 
\begin{equation*}
\mathbb{P}\left(\left\| A \right\|_{2,op}^4>t\right) \leq e^{-N\sqrt{\frac{t}{C_1^4K(N)^4}+\log(K(N))}} + e^{-N\sqrt[4]{\frac{Nt}{C_1^4K(N)^4}}}
 + e^{-\frac{Nm^2}{(64K(N)m')^2}}.
\end{equation*}

Now, we will prove that each term can be neglected.
Integrating by part, we obtain
\begin{eqnarray*}
\int_{t_1}^{t_2}  e^{-\sqrt{\frac{t}{C_0^4K(N)^4}+\log(K(N))}}\dd t &\leq& \int_{t_1}^{\infty}  e^{-N\sqrt{\frac{t}{C_0^4K(N)^4}+\log(K(N))}}\dd t 
\\ & \leq & \left[\frac{-2\sqrt{t}C_0^2K(N)^2}{N}e^{-N\sqrt{\frac{t}{C_0^4K(N)^4}+\log(K(N))}} \right]_{t_1}^{\infty}
\\ & & + \int_{t_1}^{\infty}\frac{C_0^2K(N)^2}{N\sqrt{t}}  e^{-N\sqrt{\frac{t}{C_0^4K(N)^4}+\log(K(N))}} \dd t
\\ & \leq & \frac{2\log(K(N))C_0^4K(N)^4}{N^2} + \frac{2C_0^4K(N)^4}{N^2}
\\ & = & o\left( \left(C_0K(N)\sqrt{\frac{\log(K(N))}{N}}\right)^4  \right).
\end{eqnarray*}
 
Then,
\begin{eqnarray*}
\int_{t_1}^{t_2} e^{-\sqrt[4]{\frac{Nt}{C_0^4K(N)^4}}}\dd t &\leq& t_2  e^{-N\sqrt[4]{\frac{t_1}{C_0^4K(N)^4}}}
\\ & \leq & t_2 e^{-\sqrt{N\log(K(N))}}
\\ & = & o\left( \left(C_0K(N)\sqrt{\frac{\log(K(N))}{N}}\right)^4  \right).
\end{eqnarray*}

So that,
\begin{eqnarray*}
\int_{t_1}^{t_2} e^{-x_0(N)}\dd t &\leq& t_2  e^{-\frac{Nm^2}{(64K(N)m')^2}}
\\ & = & o\left( \left(C_0K(N)\sqrt{\frac{\log(K(N))}{N}}\right)^4  \right).
\end{eqnarray*}

Leading to 
\begin{equation*}
\int_{t_1}^{t_2} \mathbb{P}\left(\left\| A \right\|_{2,op}^4>t\right) \dd t=  o\left( \left(C_0K(N)\sqrt{\frac{\log(K(N))}{N}}\right)^4  \right),
\end{equation*}
$\\$
$\\$
$\\$
Finally, for $t \in [t_2,+\infty[$, we use 
\begin{equation*}
\mathbb{P}\left(\left\| A \right\|_{2,op}^4>\max \left(\left(C_0K(N)\sqrt{\frac{\log(K(N))+x}{N}}+C_3\right)^4,\left(C_0K(N)\frac{x}{N}+C_3\right)^4\right)\right) \leq e^{-x}.
\end{equation*}

Thanks to lemma \ref{lemfun}, we get
\begin{equation*}
\mathbb{P}\left(\left\| A \right\|_{2,op}^4>t\right) \leq e^{-N\left(\frac{\sqrt[4]{t}-C_3}{C_0K(N)}\right)^2+\log(K(N))} + e^{-N\frac{(\sqrt[4]{t}-C_3)}{C_0K(N)}}.
\end{equation*}

So, integrating by part once more, we obtain
\begin{eqnarray*}
\int_{t_2}^{+\infty}  e^{-N\left(\frac{\sqrt[4]{t}-C_2}{C_0K(N)}\right)^2+\log(K(N))}\dd t &\leq&\int_{\sqrt[4]{t_2}-C_3}^{+\infty} 4(u+C_3)^3 e^{-N\left(\frac{u}{C_0K(N)}\right)^2+\log(K(N))}\dd u
\\&\leq & \left[P_1(u,N,K(N))    e^{-N\left(\frac{u}{C_0K(N)}\right)^2+\log(K(N))} \right]_{\sqrt[4]{t_2}-C_3}^{+\infty}
\\ & \leq & P_1(u,N,K(N)) e^{-N}
\\ & \leq & o\left( \left(C_0K(N)\sqrt{\frac{\log(K(N))}{N}}\right)^4  \right).
\end{eqnarray*}
Here, $P_1(u,N,K(N))$ is a polynomial of degree $3$ in $u$ and is rational function in $N$ and $K(n)$.

Furthermore,
\begin{eqnarray*}
\int_{t_2}^{+\infty} e^{-N\frac{(\sqrt[4]{t}-C_3)}{C_0K(N)}} \dd t &\leq&\int_{\sqrt[4]{t_2}-C_3}^{+\infty} 4(u+C_3)^3 e^{-N\frac{u}{C_0K(N)}}\dd u
\\&\leq & \left[P_2(u,N,K(N))    e^{-N\frac{u}{C_0K(N)}} \right]_{\sqrt[4]{t_2}-C_3}^{+\infty}
\\ & \leq & P_2(u,N,K(N))e^{-\sqrt{N(\log(K(N))+N)}}
\\ & \leq & o\left( \left(C_0K(N)\sqrt{\frac{\log(K(N))}{N}}\right)^4  \right),
\end{eqnarray*}
where $P_2(u,N,K(N))$ is a polynomial of degree $3$ in $u$ and is rational function in $N$ and $K(n)$.

We proved here 
\begin{equation*}
\int_0^\infty \mathbb{P}\left(\left\| A \right\|_{2,op}^4 > t \right) \leq C_0^4K(N)^4 (\frac{\log(K(N))}{N})^2 + o(K(N)^4 (\frac{\log(K(N))}{N})^2).
\end{equation*}
This ends the proof.
\end{proof}

\subsection{Technical lemmas}


We prove now the technical lemmas:

\begin{proof}{of Lemma~\ref{lemconc}}

Notice that $\hat{r}^{(N)}(p) = X^T T_N(g_p) X$ with $g_p(t) = \frac{N}{N-p}\cos(pt)$.
We use the following  proposition from Comte \cite{comte}. Let $X_1,\cdots,X_n$ be a centered Gaussian stationary sequence and $g$ a bounded function such that $T_n(g)$ is a symmetric non negative matrix. Then the following concentration inequality holds for $Z_n(g) = \frac{1}{n}\left(X^TT_n(g)X - \mathbb{E}[X^TT_n(g)X] \right)$:
\begin{equation*}
\mathbb{P}\left(Z_n(g) \geq 2\left\|f\right\|_\infty\left(\left\|g\right\|_2\sqrt{x} + \left\|u\right\|_\infty x\right)\right)\leq e^{-nx}.
\end{equation*}

By applying this result respectively with $g_p$ and $-g_p$ and we obtain 
\begin{equation*}
\mathbb{P}\left( \left|\hat{r}^{(N)}(p) - r(p)\right| > 2m'\frac{N}{N-p}(\sqrt{x}+x)\right) \leq 2e^{-Nx}.
\end{equation*} 
or, equivalently,
\begin{equation*}
\left|\hat{r}^{(N)}(p) - r(p)\right| > 2m'\frac{N}{N-p}(\sqrt{\frac{x+\log(K(N))+2\log(2)}{N}}+\frac{x+\log(K(N))+2\log(2)}{N}), 
\end{equation*} 
with probability lower than $\frac{e^{-x}}{2K(N)}$.
By taking an equivalent, we obtain that there exists $N_0$ such that, for all $N \geq N_0$, for all $p \leq 2K(N)$

\begin{equation*}
\mathbb{P}\left( \left|\hat{r}^{(N)}(p) - r(p)\right| > 4m'\sqrt{\frac{x+\log(K(N))}{N}} + \frac{x}{N}\right) \leq \frac{e^{-x}}{2K(N)}.
\end{equation*}
\end{proof}


\begin{proof}{of Lemma~\ref{lemfun}}

We set $t = \sup(f_1(x),f_2(x))$
If $t = f_1(x)$ then 
\begin{equation*}
\mathbb{P}\left( X > t \right) \leq e^{-f_1^{-1}(t)} + C \leq e^{-f_1^{-1}(t)}+e^{-f_2^{-1}(t)} + C.
\end{equation*}
Symmetrically, if  $t = f_2(x)$ we have 
\begin{equation*}
\mathbb{P}\left( X > t \right) \leq e^{-f_1^{-1}(t)}+e^{-f_2^{-1}(t)} + C.
\end{equation*}
\end{proof}

\begin{proof}{of Lemma~\ref{lemspec}}
It is sufficient to ensure that the bias is small enough. Choose $N_0$ such that 
\begin{equation*}
2 \left\|f^\star\right\|_{H_s}K(N)^{-s+1} \leq \frac{m}{4}.
\end{equation*}
Then we use
\begin{eqnarray*}
\left\|\hat{f}^N_{K(N)} - f^\star \right\|_\infty & \leq & \sum_{p = -K(N)}^{K(N)} \left|\hat{r}^N(p) - r(p)\right| + 2\sum_{p > K(N)}\left|r(p)\right|
\\& \leq & (2K(N)+1)\sup_{p\leq 2K(N)} \left\{ \hat{r}_N(p)-r(p)\right\} + 2 \left\|f^\star\right\|_{H_s}K(N)^{-s+1}
\\& \leq & (2K(N)+1)\sup_{p\leq 2K(N)} \left\{ \hat{r}_N(p)-r(p)\right\} +\frac{m}{4}.
\end{eqnarray*}
This ends the proof of the last lemma.
\end{proof}

\textbf{Acknowledgement :}

\normalsize
The authors
wish to thank the anonymous referee for a careful reading of the manuscript and
for providing very useful comments.


\bibliographystyle{plain}

\def\cprime{$'$}

\end{document}